\documentclass[10pt,conference,letter]{IEEEtran}
\usepackage{amsmath,amsfonts, amssymb, mathrsfs}
\usepackage[pdftex]{graphicx}
\usepackage[dvips]{color}
\usepackage[center]{caption}
\usepackage{subfigure}
\usepackage{multirow,bigdelim,longtable}
\usepackage{amsthm}
\usepackage{hhline}

\newcommand{\squeezeup}{\vspace{-2.5mm}}

\def\EE{\mathbb E}

\def\C{\mathcal{C}}
\def\N{\mathcal{N}} 

\def\V{\mathcal{V}}
\newtheorem{lemma}{Lemma}
\newtheorem{theorem}{Theorem}
\newtheorem{claim}{Claim}

\title{Improved Capacity Approximations for Gaussian Relay Networks}
\author{\IEEEauthorblockN{Ritesh Kolte and Ayfer \"{O}zg\"{u}r }
\IEEEauthorblockA{Stanford University\\
Stanford, California 94305\\
\{rkolte, aozgur\}@stanford.edu}}

\begin{document}
\maketitle
\begin{abstract}
Consider a Gaussian relay network where a number of sources communicate to a destination with the help of several layers of relays. Recent work has shown that a compress-and-forward based strategy at the relays can achieve the capacity of this network within an additive gap. In this strategy, the relays quantize their observations at the noise level and map it to a random Gaussian codebook. The resultant capacity gap is independent of the SNR's of the channels in the network but linear in the total number of nodes.

In this paper, we show that if the relays quantize their signals at a resolution decreasing with the number of nodes in the network, the additive gap to capacity can be made logarithmic in the number of nodes for a class of layered, time-varying wireless relay networks. This suggests that the rule-of-thumb to quantize the received signals at the noise level used for compress-and-forward in the current literature can be highly suboptimal.
\end{abstract}
\section{Introduction}

Consider a source node communicating to a destination node via a sequence of relays connected by point-to-point channels. See Figure~\ref{subfig:line}. The capacity of this line network is achieved by simple decode-and-forward and is equal to the minimum of the capacities of the successive point-to-point links. The decoding at each stage removes the noise corrupting the information signal and therefore the end-to-end rate achieved is independent of the number of times the message gets retransmitted. 

Unfortunately, the optimality of decode-and-forward is limited to this line topology and in more general networks with multiple relays at each layer, it is well-understood that the rate achieved by decode-and-forward can be arbitrarily away from capacity. Recent work by Avestimehr et al \cite{ADT11} has shown that compress-and-forward can be a better fit for general relay networks. In any relay network with  multi-source multicast traffic, it has been shown  that a compress-and-forward based relaying strategy can achieve the capacity of the network within a gap that is independent of the SNR's of the constituent channels \cite{ADT11,LKEC11,OD10}. However, the gap to capacity increases linearly in the number of nodes in the network. For example, for the line network in Figure~\ref{subfig:line} it would lead to a gap that is linear in the depth of the network $D$. One natural way to explain this gap is the noise accumulation. As the information signal proceeds deeper into the network, it is corrupted by more and more noise. Therefore, any strategy that does not remove the noise corrupting the signal at each stage will naturally suffer a rate loss that increases with the number of stages. However, it is not clear why this rate loss should be \emph{linear} in the depth of the network as the current results in the literature suggest \cite{ADT11,LKEC11,OD10}. The total variance of the accumulated noise over the $D$ stages of the network is $D$ times the variance of the noise at each stage (assuming identical noise variances over the $D$ stages). A factor of $D$ increase in the noise variance in a point-to-point Gaussian channel would lead to a $\log D$ decrease in capacity, and therefore it is natural to ask if we can reduce the linear performance loss of compress-and-forward strategies to logarithmic in $D$.

\begin{figure}[t]
\centering
\subfigure[]{
\includegraphics[scale=1]{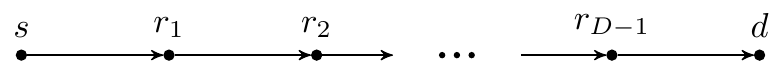}
\label{subfig:line}}
\subfigure[]{
\includegraphics[scale=1]{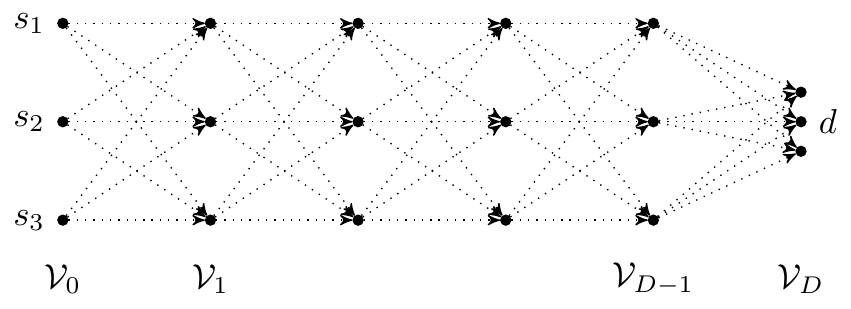}
\label{subfig:layered}}
\squeezeup\caption{(a) Line Network, (b) Multi-Layer Relay Network for $K=3$, each $H_i$ is a Rayleigh fading matrix}\label{fig:model}\squeezeup
\end{figure}
     
This paper is based on the observation that if the relay nodes in Figure~\ref{subfig:line} quantize their observed signals at a resolution decreasing linearly in $D$, the rate loss due to compress and forward is only logarithmic in $D$. (See Section \ref{sec:line}.) This suggests that the rule-of-thumb to quantize the received signals at the noise level used for compress-and-forward in the current literature \cite{ADT11,LKEC11,OD10} can be highly suboptimal. This is because the rate penalty for describing the quantized signals can be significantly larger than the rate penalty associated with coarser quantization. This insight was used in \cite{CO12} to show that compress-and-forward based strategies can achieve the capacity of the $N$-relay Gaussian diamond network within a gap that is logarithmic in $N$.

The main setup we consider in this paper is the multi-layer Gaussian relay network in Figure~\ref{subfig:layered}. Here $K$ source nodes communicate to a destination node equipped with multiple antennas over $D$ layers, each layer containing $K$ single-antenna relays. Each relay observes a noisy linear combination of the signals transmitted by the relays in the previous layer. All channels are subject to i.i.d. Rayleigh fast-fading. Current results on compress-and-forward \cite{ADT11,LKEC11,OD10} yield a sum-rate which is within $1.3\,KD$ gap to the capacity of this network, where $KD$ is the total number of nodes. Instead, we show that if relays quantize their received signals at a resolution that decreases as the number of nodes increases, compress-and-forward can achieve a sum-rate which is within an additive gap of $K\log D+K$ of the network sum-capacity. So for a fixed $K$, as the number of layers $D$ increases, this gap only grows logarithmically in the depth of the network $D$ (and therefore logarithmically in the number of nodes $KD$). As a side result, we provide an analysis of the compress-and-forward based strategies in \cite{ADT11,LKEC11,OD10} in fast fading wireless networks.

This same setup has been considered in \cite{NNW11}, where a computation alignment strategy is proposed to remove the accumulating noise  with the depth of the network. This yields a gap $7K^3 +5K\log K$. The computation alignment strategy is based on the idea of combining compute-forward \cite{NG11} with ergodic alignment proposed in \cite{NGJV09}. While the gap to capacity obtained by computation alignment is independent of $D$, this strategy is significantly more complex than compress-forward and has a number of disadvantages from a practical perspective. In particular, ergodic alignment over the fading process leads to large delays in communication and requires each relay  to know the instantaneous realizations of all the channels in the network. Moreover, its performance critically depends on the symmetry of the fading statistics. The compress-forward strategy with improved quantization we propose in this paper has minimal requirements. In particular, no channel state information is required at the source and at the relays, and the fading statistics are not critical to the operation of the strategy.

\section{Model and Preliminaries}\label{sec:model}
\subsection{Model}
We consider the configuration shown in Figure~\ref{subfig:layered}. The network is a directed layered network, each layer except the last containing $K$ nodes. The nodes in the $i$th layer are collectively referred to as $\V_i$ where $0\leq i\leq D$. Nodes in $\V_0$ are the $K$ source nodes $\{s_j\}_{j=1}^K$, having messages at rate $R_j$ to be communicated to the single destination node $d$ in $\V_D$, which has $K$ antennas. Since $\V_D$ only contains $d$, we use $d$ and $\V_D$ interchangeably in the sequel. We assume that $d$ is equipped with multiple antennas in order to keep the problem interesting. Otherwise, the minimum cut becomes the multiple-input-single-output cut from the last layer of relays to $d$ and this trivializes the problem of approximately achieving the capacity of the network. Instead of multiple antennas at $d$, one can also assume orthogonal bit-pipes from nodes in $\V_{D-1}$ to $d$, as done in \cite{NNW11}. Let $\V^{i}$ denote $\V_0\cup\V_1\cup\dots\cup \V_i$ and $\N$ denote the set of all nodes, i.e. $\N=\V^D$. 

For $0\leq i\leq D-1$, the received signal at nodes in $\V_{i+1}$ (or antennas if $i=D-1$) depends only on the transmit signals of nodes in $\V_i$ and at time $t$ is given by
$$Y_{\V_{i+1}}[t]=H_{\V_i\rightarrow\V_{i+1}}[t]X_{\V_i}[t]+Z_{\V_{i+1}}[t],$$
where $Y_{\V_{i+1}}$ and $X_{\V_i}$ are vectors containing the received and transmitted signals at nodes in $\V_{i+1}$ and $\V_i$ respectively; and $Z_{\V_{i+1}}\sim\C\N(0,\sigma^2I)$, i.e. we assume flat-fading channels between the nodes with i.i.d. circularly symmetric complex Gaussian noise. The $(k,l)$'th entry of the matrix $H_{\V_i\rightarrow\V_{i+1}}[t]$ denotes the channel coefficient from $l$'th relay  in $\V_i$ to $k$'th relay in $\V_{i+1}$ at time $t$. We further assume that channels are i.i.d. Rayleigh fading, i.e  each entry in the matrices $\{H_{\V_0\rightarrow\V_1}[t],H_{\V_1\rightarrow\V_2}[t],\dots,H_{\V_{D-1}\rightarrow d}[t]\}$ is i.i.d. $\C\N(0,1)$ across time, and independent of other entries and independent of the noise and transmissions. (The conclusions of the paper also hold under a block fading model.) All transmitting nodes are subject to a long-term average power constraint $P$. We can assume that $Y_{\V_0}=0$ and $X_{d}=0$. The source nodes and the relay nodes do not know the instantaneous realizations of the channel coefficients, i.e have no transmit or receive channel state information. (The source nodes know the topology of the network and the channel statistics, i.e. the end-to-end ergodic rate supported by the network.) All channel realizations are known at the destination node and are used while decoding the transmitted messages from the source nodes. The largest achievable sum-rate $\sum_{j=1}^{K}R_j$ in the network is called the sum-capacity of the network, denoted by $C_{sum}$.

\subsection{Preliminaries}
A cut $\Omega$ is a subset of $\N$. %
Let $H[t]$ be a random vector containing all the channel realizations in the network. Since the channel realizations are known at the destination $d$, we can view $H[t]$ as part of the output of $d$ at time $t$, i.e., at time $t$, $d$ observes $Y_d[t]$ and $H[t]$. Note that this does not alter the memorylessness property of the network. For the sake of notational convenience in the proofs, we define the following quantity for a cut $\Omega$,
\begin{IEEEeqnarray}{rCl}
\overline{C}(\Omega) & := & I(X_{\Omega};Y_{\Omega^c},H|X_{\Omega^c})\nonumber\\
& = & I(X_{\Omega};Y_{\Omega^c}|X_{\Omega^c},H) \label{eq:cutset_expr}\end{IEEEeqnarray} where $X_{\N}$ are jointly distributed with some distribution such that the average power constraints are satisfied. The second equality follows from the fact that $I(X_{\Omega};H|X_{\Omega^c})=0$ since the distribution of $X_{\mathcal{N}\setminus d}$ is independent of $H$ ($H$ is unknown to all nodes but the destination) and $X_{\V_d}=0$. With this notation, the information-theoretic cutset upper bound \cite[Theorem 15.10.1]{CT06} on the achievable rates in the network can be expressed as follows: If each source $s_j$ can reliably communicate at a rate $R_j$ simultaneously, then there exists some joint distribution $p(X_{\N})$ on $X_{\N}$ such that 
\begin{equation}\label{eq:cutset}\sum_{j:s_j\in\Omega, d\in\Omega^c} R_j\leq \overline{C}(\Omega)\quad\quad \text{for all cuts }  \Omega.\end{equation}

\section{Line Network}\label{sec:line}
We first illustrate the main idea of this paper in a simple setting, the line network in Figure~\ref{subfig:line}. Here we assume that each link $i$ is a AWGN channel with gain $h_i$ and the channel gains $h_i$ are fixed and known. Each node has power $P$ and the noise variance is $\sigma^2$. (The conclusions below also hold under a fast fading assumption similar to the one described in Section~\ref{sec:model}.)  As mentioned before, a decode-forward strategy at the relays achieves the capacity of this line network, while compress-and-forward based strategies (such as quantize-map-forward in \cite{ADT11} and noisy network coding in  \cite{LKEC11}) with quantization done at the noise level have a gap to capacity that is linear in the number of nodes $D$. Here, we show that if relays instead quantize at $(D-1)$ times the noise level, the gap to capacity becomes logarithmic in $D$. 

Number the nodes $s$ through $d$ as $0,1,2,\dots,D$. Let's consider the rate achievable by noisy network coding  for this network, assuming all relay nodes choose their transmission codebooks independently from a Gaussian distribution, i.e. $X_i\sim\C\N(0,P)$ and independent of each other. Theorem 1 in \cite{LKEC11} says that the following rate is achievable
$$R=\min_{0\leq i\leq D-1} \left(I(X_i;\hat{Y}_{i+1}|X_{i+1})-I(Y_{\V^i};\hat{Y}_{\V^i}|X_{\N},\hat{Y}_{\N\setminus\V^i})\right)$$
where we are assuming that the destination node also performs quantization for simplicity. %

Now, let each relay choose $\hat{Y}_i=Y_i+\hat{Z}_i$ where $\hat{Z}_i\sim\N(0,(D-1)\sigma^2)$ independent of everything else. Since $Y_{i+1}=h_iX_{i}+Z_{i+1}$, the channel from $X_i$ to $\hat{Y}_{i+1}$ is effectively an AWGN channel of noise power $D\sigma^2$ with gain $h_i$. Then the first term in the achievable rate expression becomes $\log\left(1+\frac{|h_i|^2P}{D\sigma^2}\right)$ which is greater than or equal to $\log\left(1+\frac{|h_i|^2P}{\sigma^2}\right)-\log(D)$.

Due to the coarse quantization, the second term in the achievable rate expression is reduced significantly as compared to quantizing at the noise level.  We have
\begin{align*}
I(Y_{\V^i};\hat{Y}_{\V^i}|X_{\N},\hat{Y}_{\N\setminus\V^i})&=I(Z_{\V^i};\hat{Z}_{\V^i})\\
&=(|\V^i|-1)\log\left(1+\frac{\sigma^2}{(D-1)\sigma^2}\right)\\
&=i\log\left(1+\frac{\sigma^2}{(D-1)\sigma^2}\right) \\
&\leq \frac{i}{D-1} \leq 1,
\end{align*}
since $i\leq D-1.$ Since the capacity of the line network is given by the minimum of the capacities of each link: $\min_i \log(1+|h_i|^2P)$, we see that decreasing the resolution of quantization as the number of nodes increases results in a gap of $\log(D)+1.$ If the quantization were done at the noise level, the first term in the noisy network coding achievable rate would  suffer from only a $\log(2)$ decrease instead of $\log(D)$ with respect to capacity, however the second term would be linear in $D$, overall resulting in a linear gap in $D$ to capacity.

\section{Layered Network with Multiple Relays}\label{sec:layered}
The main result of this paper is the following theorem.

\begin{theorem}\label{thm:main}
The sum-capacity of the network in Figure~\ref{subfig:layered} is bounded by 
\begin{equation}\label{eq:mainres}
C(K,K)-K\log (D)-K \leq C_{sum} \leq C(K,K)
\end{equation}
where the lower bound is achievable by a compress-and-forward strategy with appropriately chosen quantization levels. $C(K,K)$ denotes the ergodic capacity of a $K$-by-$K$ MIMO Rayleigh fast-fading channel with per-antenna average power constraint of $P$ and noise variance $\sigma^2$ and is equal to the information-theoretic cutset upper bound on the sum-capacity of the network. 
\end{theorem}

We first prove Theorem~\ref{thm:main} for the case when the $K$ source nodes $\{s_1,\dots,s_K\}$ are co-located, i.e. $\{s_1,\dots,s_K\}$ behave like a single source denoted by $s$ with $K$ antennas, with a per-antenna power constraint $P$, transmitting a message at rate $R$ to the destination $d$, see Figure~\ref{fig:cut}. In this case, we show that the point-to-point capacity satisfies the conditions in Theorem~\ref{thm:main}. At the end of this section, we extend the proof for the capacity in the single-source case to the sum-capacity in the original setup containing multiple sources.

We prove Theorem~\ref{thm:main} for the single source setup in two steps. We first establish the upper bound on the capacity in Section~\ref{subsec:upper} and then show that it is achievable within a gap $K\log (D)+K$ in Section~\ref{subsec:nnc}.

\subsection{Upper bound}\label{subsec:upper} 

The upper bound in Theorem \ref{thm:main} is easy to prove. Consider the cutset upper bound in \eqref{eq:cutset} for the single source case:
$$R\leq \min_{\Omega : s\in\Omega,d\in\Omega^c} I(X_{\Omega};Y_{\Omega^c}|X_{\Omega^c},H).$$
Considering only the cut $\Lambda=\V_0$ implies that
\begin{IEEEeqnarray}{rCl}
R & \leq & \max_{p(X_{\N})}\min_{\Omega :s\in\Omega,d\in\Omega^c} \overline{C}(\Omega)\nonumber\\%&\leq & \begin{array}{ccc} \text{min of all bounds on}\\ \text{sum-rate obtained from \eqref{eq:cutset}}\end{array}\nonumber\\
& \leq & \max_{p(X_{\N})}\overline{C}(\V_0)\label{eq:tight_ineq}\\
&=&\max_{p(X_{\N})}I(X_{\V_0};Y_{\N\setminus\V_0}|X_{\N\setminus\V_0},H)\nonumber\\
&\stackrel{(a)}{=}&\EE\left[\log\det\left(I+\frac{1}{\sigma^2}PH_{\V_0\rightarrow\V_1}H_{\V_0\rightarrow\V_1}^{\dagger}\right)\right]\nonumber\\
&\triangleq &C(K,K)\nonumber,
\end{IEEEeqnarray}
where (a) follows from the fact that the maximal mutual information in the earlier line corresponds to the ergodic capacity of a $K\times K$  MIMO Rayleigh fast-fading  channel with per-antenna average power constraint $P$ and the maximizing input distribution  for this channel is well known to be i.i.d. $\C\N(0,P)$ \cite{T99}. We denote this capacity by $C(K,K)$.

\paragraph*{Remark} The cutset upper bound in \eqref{eq:cutset} bounds the rate with many additional constraints arising from cuts other than $\V_0$. In the above derivation, by concentrating on a single cut $\Lambda=\V_0$ we have derived an upper bound \eqref{eq:tight_ineq} on the cutset bound. Although such an upper bound can be weaker in general, in the current case it can be shown that $C(K,K)$ is indeed the tightest constraint on the rate imposed by the cutset bound. This can be observed from the discussion in the next section (Claim~\ref{claim1}), which implicitly shows that the cutset bound on the rate evaluated under i.i.d distributions is equal to $C(K,K)$. Since the cutset bound evaluated under a particular distribution forms a lower bound on the actual bound obtained from \eqref{eq:cutset}, this shows that the tightest constraint imposed on the rate by the cutset bound is exactly equal to $C(K,K)$.

\subsection{Achievability}\label{subsec:nnc}
We now prove the lower bound in Theorem~\ref{thm:main}. We start with the rate achieved by noisy network coding in \cite[Theorem~1]{LKEC11}, which states that all rates $R$ that satisfy  %
\begin{IEEEeqnarray}{lCl}
R & \leq & \min_{\Omega :s\in\Omega ,d\in\Omega^c} \left[ I(X_{\Omega};\hat{Y}_{\Omega^c},H|X_{\Omega^c})\right.\nonumber\\
&& \quad\quad\quad\quad\quad\quad\quad \left. -I(Y_{\Omega};\hat{Y}_{\Omega}|X_{\N},\hat{Y}_{\Omega^c},H)\right]\nonumber\\
& = & \min_{\Omega :s\in\Omega ,d\in\Omega^c} \left[ I(X_{\Omega};\hat{Y}_{\Omega^c}|X_{\Omega^c}, H)\right.\nonumber\\
&& \quad\quad\quad\quad\quad\quad\quad \left. -I(Y_{\Omega};\hat{Y}_{\Omega}|X_{\N},\hat{Y}_{\Omega^c},H)\right].\nonumber
\end{IEEEeqnarray}
for some joint distribution of the form $\Pi_{k\in \N} p(x_k)p(\hat{y}_k| y_k, x_k)$ are achievable. The equality again follows from the fact that $I(X_{\Omega};H|X_{\Omega^c})=0$. Hence, the following $R$ is achievable:
\begin{IEEEeqnarray}{lCl}
R & \leq & \min_{\Omega :s\in\Omega ,d\in\Omega^c} I(X_{\Omega};\hat{Y}_{\Omega^c},H|X_{\Omega^c})\nonumber\\
&& \quad\quad\quad\quad  -\max_{\Omega :s\in\Omega ,d\in\Omega^c}I(Y_{\Omega};\hat{Y}_{\Omega}|X_{\N},\hat{Y}_{\Omega^c},H).\label{eq:nnc_ach}
\end{IEEEeqnarray}
We choose the input distribution $X_k$ at each node $k$ to be $\C\N(0,P)$ (and similarly the input distributions corresponding to the antennas of the source node are i.i.d. $\C\N(0,P)$).  We choose $\hat{Y}_k$ such that $\hat{Y}_k=Y_k + \hat{Z}_k$ where $\hat{Z}_k$ is $\C\N(0,(D-1)\sigma^2)$ independent of everything else. Note the difference with the quantization in \cite{ADT11,LKEC11,OD10}: the quantization noise has variance $(D-1)\sigma^2$ as opposed to $\sigma^2$, the noise variance. For simplicity, we also assume that the destination quantizes its observation according to $\hat{Y}_{\V_D}=Y_{\V_D}+\hat{Z}_{\V_D},$  where $\hat{Z}_{\V_D}\sim\C\N(0,(D-1)\sigma^2)$ independent of everything else, and treats $\hat{Y}_{\V_D}$ (denoted by $\hat{Y}_d$ for brevity) as its observation, along with all the channel realizations $H$.

We will evaluate the right-hand side of  \eqref{eq:nnc_ach} in two steps. In Lemma~\ref{lem:nnc}, we upper bound the second term by $K$. In Lemma~\ref{lem:iid_mincut}, we lower bound the first term by $C(K,K)-\log D$. Combining the two results gives the lower bound in Theorem~\ref{thm:main} (for the single source case).

\begin{lemma}\label{lem:nnc}

$$\max_{\Omega :s\in\Omega ,d\in\Omega^c}I(Y_{\Omega};\hat{Y}_{\Omega}|X_{\N},\hat{Y}_{\Omega^c},H)\leq K$$
\end{lemma}
\begin{proof}
Given our choice for the distributions of the random variables involved, we have
\begin{IEEEeqnarray*}{lCl}
I(Y_{\Omega};\hat{Y}_{\Omega}|X_{\N},\hat{Y}_{\Omega^c},H)\\  
\quad = h(\hat{Y}_{\Omega}|X_{\N},\hat{Y}_{\Omega^c},H)-h(\hat{Y}_{\Omega}|Y_{\Omega},X_{\N},\hat{Y}_{\Omega^c},H)\\
\quad \leq h(\hat{Y}_{\Omega}|X_{\N},H)-h(\hat{Y}_{\Omega}|Y_{\Omega},X_{\N},H)\\
\quad = (|\Omega |-1)\log(D\sigma^2)-(|\Omega |-1)\log ((D-1)\sigma^2) \\
\quad \leq K(D-1)\log\left(1+\frac{1}{D-1}\right)\leq K.
\end{IEEEeqnarray*}
Hence $\max_{\Omega :s\in\Omega ,d\in\Omega^c}I(Y_{\Omega};\hat{Y}_{\Omega}|X_{\N},\hat{Y}_{\Omega^c},H)\leq K.$
\end{proof}

We next lower bound the first term in \eqref{eq:nnc_ach}. 

\begin{lemma}\label{lem:iid_mincut}
$$\min_{\Omega :s\in\Omega, d\in\Omega^c} I(X_{\Omega};\hat{Y}_{\Omega^c}|X_{\Omega^c},H) \geq C(K,K)-K\log D $$
\end{lemma}

\begin{proof}
We first prove the following relation:

\begin{claim}\label{claim1}
$$\min_{\Omega :s\in\Omega, d\in\Omega^c} I(X_{\Omega};Y_{\Omega^c}|X_{\Omega^c},H) = C(K,K).$$
\end{claim}

\begin{figure}[t]
\centering
\includegraphics[scale=1]{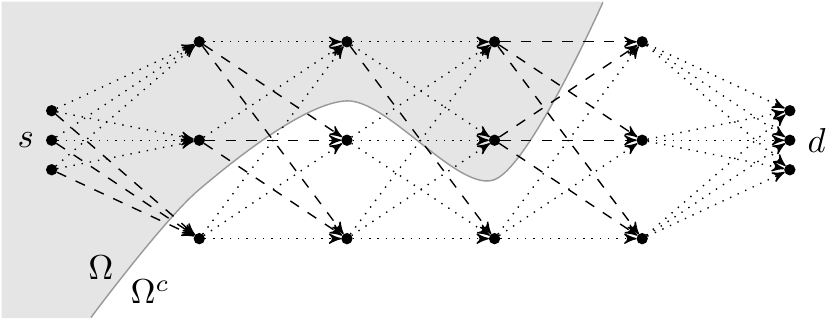}\squeezeup
\caption{Links crossing the cut $\Omega$ denoted by dashed lines; $M_1=2,M_2=1,M_3=2,M_4=0$}\label{fig:cut}
\end{figure}

For notational convenience in this proof, we define $C_{i.i.d.}(\Omega):=I(X_{\Omega};Y_{\Omega^c}|X_{\Omega^c},H)$, where we emphasize that the inputs are i.i.d. $\C\N(0,P)$ via the subscript ``\emph{i.i.d.}''.

Consider a cut $\Omega$ that contains $M_1$ nodes from $\V_1$, $M_2$ from $\V_2$ and so on until $M_{D-1}$ from $\V_{D-1}$ (see Figure~\ref{fig:cut}). Recall that we assume $s\in\Omega$ and $d\in\Omega^c$. Then $C_{i.i.d.}(\Omega)$ is given by $$\EE\left[\log\det\left(I+\frac{P}{\sigma^2}H_{\Omega\rightarrow\Omega^c}H_{\Omega\rightarrow\Omega^c}^{\dagger}\right)\right],$$
where $H_{\Omega\rightarrow\Omega^c}$ is a block diagonal matrix containing blocks of size $M_1^c$-by-$K$, $M_2^c$-by-$M_1$, $M_3^c$-by-$M_2$, $\dots$, $M_{D-1}^c$-by-$M_{D-2}$ and finally $K$-by-$M_{D-1}$. We have abused notation here by defining $M_i^c:=|\V_i|-M_i=K-M_i.$

Since the capacity of a MIMO channel that has block diagonal structure is the sum of the capacities of the individual MIMO blocks, we have
\begin{IEEEeqnarray*}{lCl} 
C_{i.i.d.}(\Omega) & =&  \EE\left[\log\det\left(I+\frac{P}{\sigma^2}H_{\Omega\rightarrow\Omega^c}H_{\Omega\rightarrow\Omega^c}^{\dagger}\right)\right]\\
& = & C(M_1^c,K)+C(M_2^c,M_1)+C(M_3^c,M_2)+\\
&&\quad\dots +C(M_{D-1}^c,M_{D-2})+C(K,M_{D-1}) \quad (*)
\end{IEEEeqnarray*}
We show below that $(*)\geq C(K,K)$. Note the following properties of the function $C(x,y)$:
\begin{itemize}
\item[a)] $C(x,y)=C(y,x)$,
\item[b)] $C(z,y)\geq C(x,y)$ if $z\geq x$,
\item[c)] $C(x,y)+C(K-x,y)\geq C(K,y)$ which can be shown via an application of Hadamard's inequality.
\end{itemize}
Proving that the expression in $(*)\geq C(K,K)$ is just a matter of applying these properties multiple times. For concreteness, we show this for the case $D=4$ below, which can be generalized in a straightforward way to higher values of $D$.
\begin{IEEEeqnarray*}{rCl}
(*) & = & C(M_1^c,K)+C(M_2^c,M_1)+C(M_3^c,M_2)+C(K,M_3)\\
& \geq & C(M_1^c,K)+C(M_2^c,M_1)+C(M_3^c,M_2)+C(M_2,M_3)\\
& \geq & C(M_1^c,K)+C(M_2^c,M_1)+C(K,M_2)\\
& \geq & C(M_1^c,K)+C(M_2^c,M_1)+C(M_1,M_2)\\
& \geq & C(M_1^c,K)+C(K,M_1)\\
& \geq & C(K,K),
\end{IEEEeqnarray*} 
where the first inequality follows by applying property (b) to the last term in the first line, the second inequality follows by applying (c) to the last two terms in the earlier line etc. So we have shown that \begin{equation}\label{eq:lem1}\min_{\Omega:s\in\Omega,d\in\Omega^c}C_{i.i.d.}(\Omega)\geq C(K,K).\end{equation} The cuts $\V^0,\V^1,\dots,\V^{D-1}$ satisfy \eqref{eq:lem1} with equality, so we are done. (Each of these cuts induces a $K$-by-$K$ MIMO channel across the cut, i.e. $C_{i.i.d.}(\V^i)=C(K,K)$ for any $0\leq i\leq D-1$.) This proves Claim~\ref{claim1}.

Due to our choice of the quantization: $\hat{Y}=Y+\hat{Z}$ where $\hat{Z}\sim\C\N(0,D-1)$, evaluating the term $I(X_{\Omega};\hat{Y}_{\Omega^c}|X_{\Omega^c},H)$ is equivalent to evaluating $I(X_{\Omega};Y_{\Omega^c}|X_{\Omega^c},H)$ except that now the noise is $Z+\hat{Z}$ instead of just $Z$, i.e. the noise power is $D\sigma^2$ instead of $\sigma^2$. %
Hence, %
\begin{IEEEeqnarray}{lCl}
\min_{\Omega:s\in\Omega ,d\in\Omega^c} I(X_{\Omega};\hat{Y}_{\Omega^c}|X_{\Omega^c},H)\nonumber\\
\quad\quad =\EE\left[\log\det\left(I+\frac{P}{D\sigma^2}H_{\V_i\rightarrow\V_{i+1}}H_{\V_i\rightarrow\V_{i+1}}^{\dagger}\right)\right]\nonumber\\
\quad\quad \geq C(K,K) - K\log(D)\label{eq:nnc_lem2},
\end{IEEEeqnarray}

This concludes the proof of the lemma.
\end{proof}

\vspace{3mm}

\subsection{Proof of Theorem~\ref{thm:main}}

Via Lemma~\ref{lem:iid_mincut} and Lemma~\ref{lem:nnc}, we have proved Theorem~\ref{thm:main} for the case of a single $K$-antenna source. We now show that the same result holds for the sum-capacity in the original setup containing $K$ single-antenna sources.

It is clear that the upper bound established in Section~\ref{subsec:upper} is an upper bound on the achievable sum-rate for the $K$ sources. 

For the lower bound, we observe that since in the above discussion  we have chosen i.i.d. input distributions for the antennas at the source, we can apply the same strategy and therefore achieve the same total rate even if antennas are not collocated. A more formal argument can be made as follows. %
Consider the setup with $K$ sources as shown in Figure~\ref{subfig:layered}. We fix the operation of the relays to be the same as that described in Section~\ref{subsec:nnc}. This induces a multiple access channel between the sources $s_1,s_2,\dots,s_K$ and the destination $d$ described by a certain pdf $p\left(y_d,H|x_{\V_0}\right)$.  It is well known that the achievable rate region for a memoryless MAC channel is a polymatroid and the largest achievable sum-rate is given by $I(X_{\V_0};Y_{d},H)$  where $p(x_{\V_0})=\prod_{j=1}^K p(x_{s_j})$ since the transmitting nodes can cannot cooperate. If we fix $p(x_{s_j})$ to be the $\C\N(0,P)$ pdf for all $1\leq j\leq K$, then $I(X_{\V_0};Y_{d},H)$ is the same end-to-end mutual information that we obtain in the case of a single source with $K$ antennas using the achievability scheme in Section~\ref{subsec:nnc}. Thus, the lower bound on the capacity for a single source with $K$ antennas that we proved in Lemma~\ref{lem:iid_mincut} also applies to the sum-capacity in the case of $K$ single antenna sources. This completes the proof of Theorem~\ref{thm:main}.
\hfill $\QED$

\paragraph*{Remark} We point out that Theorem~\ref{thm:main} continues to hold if there are multiple destination nodes in the final layer, each having $K$ antennas and interested in all the messages.

\section{Concluding Remarks}
In this paper, we have considered a time-varying Gaussian relay network in which $K$ sources communicate to a destination over multiple layers of relays, each layer containing $K$ nodes. We have shown that by better choosing the quantization level in the compress-and-forward strategies, we can improve the gap to capacity from linear to logarithmic in the depth of the network. This is obtained by decreasing the resolution of quantization as the number of nodes in the network increases, which decreases the associated rate penalty to communicate the quantization codewords to the destination.

\bibliographystyle{IEEEtran}
\bibliography{IEEEabrv,itw_gap}
\end{document}